\documentclass[12pt]{article}%
\usepackage{mathrsfs}
\usepackage{amsfonts}
\usepackage{amsthm}
\usepackage{amssymb}
\usepackage{amsmath}
\usepackage[english]{babel}
\usepackage[utf8]{inputenc}
\usepackage[OT2,T1]{fontenc}
\usepackage[numbers,sort&compress]{natbib}
\usepackage{graphicx}%
\providecommand{\U}[1]{\protect\rule{.1in}{.1in}}
\addto{\captionsenglish}{}

\newtheorem{thm}{Theorem}

\begin{document}

\title{\protect\vspace{-1.6cm}Dispersionless (3+1)-dimensional integrable hierarchies}
\author{Maciej B\l aszak$^{a}$ and Artur Sergyeyev$^{b}$\\$^{a}$~Faculty of Physics, A. Mickiewicz University, \\Umultowska 85, 61-614 Pozna\'n, Poland\\\texttt{blaszakm@amu.edu.pl}\\[3mm] $^{b}$~Mathematical Institute, Silesian University in Opava,\\Na Rybn\'\i{}\v{c}ku 1, 746 01 Opava,~Czech Republic\\\texttt{Artur.Sergyeyev@math.slu.cz}}
\maketitle

\begin{abstract}\protect\vspace{-1.2cm}
In the present paper we introduce a multi-dimensional version of the
$R$-matrix approach to the construction of integrable hierarchies. Applying
this method to the case of the Lie algebra of functions with respect to the contact bracket, we construct integrable
hierarchies of (3+1)-dimensional dispersionless systems of the type recently
introduced by one of us in \cite{s}.\vspace{-0.3cm}
\end{abstract}

\section{Introduction}

Integrable systems are well known to play a prominent role in modern theoretical and mathematical physics,
including quantum field theory and string theory, cf.\ e.g.\ \cite{a, ba, bl,
bla1, bo, br, c, dun, sem2, w, tu}. The $R$-matrix approach, see e.g.\ \cite{bl, bla1,
sem2} and references therein, is one of the most general and best known
constructions of such systems. In this approach integrable systems result from
the Lax equations on suitably chosen Lie algebras. The key advantage of this
method is the possibility of systematic construction of infinite hierarchies
of symmetries, conserved quantities and respective Hamiltonian, or rather
multi-Hamiltonian, structures, see e.g.\ the recent surveys \cite{bla1,sem2}.\looseness=-1

More than three decades of experience show that this approach, as well as
other methods, works perfectly in (1+1) dimensions and admits an extension to
(2+1) dimensions, see for example \cite{bl, bla2, bla3, t, F1,F2}. However, to
the best of our knowledge, all earlier attempts at extending these methods to
higher dimensions failed. In particular, until recently only isolated examples
of (3+1)-dimensional integrable partial differential systems were known,
cf.\ e.g.\ \cite{dun} and references therein.

A significant advance in this direction was made in \cite{s},
where a novel systematic construction of (3+1)-dimensional
integrable dispersionless systems was found. To explain this
construction, recall that zero-curvature equations involving the
Poisson bracket with one degree of freedom give rise to
(2+1)-dimensional dispersionless systems, see for example
\cite{bla3,s,sz,z}. Roughly speaking, the key insight of \cite{s}
is to replace the Poisson bracket by the contact bracket in the
zero-curvature equations in question. Then these equations yield
(3+1) rather than (2+1)-dimensional systems. This approach gives
rise to broad new classes of (3+1)-dimensional dispersionless
integrable systems along with their Lax pairs.\looseness=-1

Motivated by the results of \cite{s}, we present below a multi-dimensional
version of the $R$-matrix approach on appropriately chosen Lie algebras. In
contrast with the standard version of the $R$-matrix method, we drop the requirement that the Lie
algebras under study admit, in addition to the Lie bracket, an associative
multiplication such that the adjoint action associated with the Lie bracket is
a derivation (that is, this action obeys the Leibniz rule) with respect to the
said multiplication.
Unfortunately, in this case there appears to be no natural Hamiltonian
structure on the dual Lie algebra, and thus no systematic method for
constructing Hamiltonian representations for the systems under study is available.

In the particular setting introduced in \cite{s} and considered in
Section~\ref{(3+1)}, the Lie algebras belong to the class of Jacobi algebras
which represent a natural generalization of the Poisson algebras. Even though
the Jacobi algebras by definition admit an associative multiplication in
addition to the Lie bracket, the adjoint action associated with the Lie
bracket is not a derivation; instead it obeys a certain generalization of the
Leibniz rule presented in Section~\ref{scb}. The systems in question are
integrable in the sense of existence of infinite hierarchies of commuting
symmetries, and the construction of these hierarchies is given below. Note
also that infinite hierarchies of nonlocal conservation laws for the systems
under study could be obtained using the construction of nonisospectral Lax pairs from
\cite{s} applied to our systems.\looseness=-1

Using the $R$-matrix approach with suitably relaxed assumptions presented in
Section~\ref{rma}, in Section~\ref{(3+1)} we construct infinite hierarchies of
integrable dispersionless (3+1)-dimensional systems with infinitely many
dependent variables associated with the contact bracket which is discussed in
Section~\ref{scb}. Finally, some natural finite-component reductions of our systems are
presented in Section~\ref{fc}.

\section{The general $R$-matrix construction of integrable hierarchies}

\label{rma}

Let $\mathfrak{g}$ be an (infinite-dimensional) Lie algebra. The Lie bracket
$[\cdot,\cdot]$ defines the adjoint action of $\mathfrak{g}$ on $\mathfrak{g}%
$: $\operatorname{ad}_{a}b=[a,b]$.

Recall, see e.g.\ \cite{sem, bla1} and references therein, that an
$R\in\mathrm{End}(\mathfrak{g})$ is called a (classical) $R$-matrix if the
$R$-bracket
\begin{equation}
\lbrack a,b]_{R}:=[Ra,b]+[a,Rb] \label{2.2}%
\end{equation}
is a new Lie bracket on $\mathfrak{g}$. The skew symmetry of (\ref{2.2}) is
obvious. As for the Jacobi identity for (\ref{2.2}), a sufficient condition
for it to hold is the so-called classical modified Yang--Baxter equation for
$R$,
\begin{equation}
\lbrack Ra,Rb]-R[a,b]_{R}-\alpha\lbrack a,b]=0,\qquad\alpha\in\mathbb{R}.
\label{2.3}%
\end{equation}

Let $L_{i}\in\mathfrak{g}$, $i\in\mathbb{N}$. Consider the associated
hierarchies of flows (Lax hierarchies)
\begin{equation}
(L_{n})_{t_{r}}=[RL_{r},L_{n}],\qquad r,n\in\mathbb{N}. \label{2.5}%
\end{equation}
We have the following

\begin{thm}
\label{t1} Suppose that $R$ is an $R$-matrix on $\mathfrak{g}$ which commutes
with all derivatives $\partial_{t_{n}}$, i.e.,
\begin{equation}
(RL)_{t_{n}}=RL_{t_{n}},\quad n\in\mathbb{N}, \label{2.6}%
\end{equation}
and obeys the classical modified Yang--Baxter equation (\ref{2.3}) for
$\alpha\neq0$. Let $L_{i}\in\mathfrak{g}$, $i\in\mathbb{N}$ satisfy (\ref{2.5}).

Then the following conditions are equivalent:

\begin{itemize}
\item[i)] the zero-curvature equations%
\begin{equation}
(RL_{r})_{t_{s}}-(RL_{s})_{t_{r}}+[RL_{r},RL_{s}]=0,\quad r,s\in\mathbb{N}
\label{2.7}%
\end{equation}
hold;

\item[ii)] all $L_{i}$ commute in $\mathfrak{g}$:
\begin{equation}
[L_{i},L_{j}]=0, \qquad i,j\in\mathbb{N}. \label{2.4}%
\end{equation}

\end{itemize}

Moreover, if one (and hence both) of the above equivalent conditions holds,
then the flows (\ref{2.5})
commute, i.e.,
\begin{equation}
((L_{n})_{t_{r}})_{t_{s}}-((L_{n})_{t_{s}})_{t_{r}}=0, \quad n,r,s\in
\mathbb{N}. \label{2.8}%
\end{equation}

\end{thm}

\begin{proof}
Using (\ref{2.5}) and the assumption (\ref{2.6}) we see that the
left-hand side of (\ref{2.7}) takes the form
\begin{align*}
&  (RL_{r})_{t_{s}}-(RL_{s})_{t_{r}}+[RL_{r},RL_{s}]\\
&  =R[RL_{s},L_{r}]-R[RL_{r},L_{s}]+[RL_{r},RL_{s}]\\
&  =[RL_{r},RL_{s}]-R[L_{r},L_{s}]_{R}\stackrel{(\ref{2.3})}{=}-\alpha\lbrack
L_{r},L_{s}]
\end{align*}
which establishes the equivalence of (\ref{2.4}) and (\ref{2.7}).
To complete the proof it suffices to observe that 
the left-hand side of (\ref{2.4}) can be written as
\[
\begin{array}{rcl}
((L_{n})_{t_{r}})_{t_{s}}-((L_{n})_{t_{s}})_{t_{r}} &=&[RL_{r},L_{n}]_{t_{s}%
}-[RL_{s},L_{n}]_{t_{r}}\\
&=&[(RL_{r})_{t_{s}}-(RL_{s})_{t_{r}},L_{n}]+[RL_{r},[RL_{s},L_{n}]]\\
&&-[RL_{s},[RL_{r},L_{n}]]\\
&=&[(RL_{r})_{t_{s}}-(RL_{s})_{t_{r}}+[RL_{r},RL_{s}],L_{n}]\\
&=&0,
\end{array}
\]
where the last equality follows from (\ref{2.7}).
\end{proof}

Now we present a procedure of extending the systems under study by adding an
extra independent variable.
This procedure bears some resemblance to that of central extension, see
e.g.\ \cite{bla1,bla2,sem2} and references therein.\looseness=-1

Namely, we assume that all elements of $\mathfrak{g}$ depend on an additional
independent variable $y$ not involved in the Lie bracket, so all of the above
results
remain valid. Consider
an $\mathcal{L}\in\mathfrak{g}$ and the associated Lax hierarchies
defined by
\begin{equation}
\mathcal{L}_{t_{r}}=[RL_{r},\mathcal{L}]+(RL_{r})_{y},\qquad r\in\mathbb{N}.
\label{2.12}%
\end{equation}

\begin{thm}
\label{t2} Suppose that $\mathcal{L}\in\mathfrak{g}$ and $L_{i}\in
\mathfrak{g}$, $i\in\mathbb{N}$ are such that the zero-curvature equations
(\ref{2.7}) hold for all $r,s\in\mathbb{N}$, the $R$-matrix $R$ on
$\mathfrak{g}$ satisfies (\ref{2.6}), and equations (\ref{2.12}) hold for all
$r\in\mathbb{N}$.

Then the flows (\ref{2.12}) commute, i.e.,
\begin{equation}
(\mathcal{L}_{t_{r}})_{t_{s}}-(\mathcal{L}_{t_{s}})_{t_{r}}=0,\quad
r,s\in\mathbb{N}. \label{2.13}%
\end{equation}

\end{thm}

\begin{proof}
Using equations (\ref{2.12}) and the Jacobi identity for the Lie bracket we obtain
\begin{eqnarray*}
(\mathcal{L}_{t_{r}})_{t_{s}}-(\mathcal{L}_{t_{s}})_{t_{r}} &=&\left[ (RL_{r})_{t_{s}}-(RL_{s})_{t_{r}}+[RL_{r},RL_{s}],\mathcal{L}\right]  \\
&&+\left( (RL_{r})_{t_{s}}-(RL_{s})_{t_{r}}+[RL_{r},RL_{s}]\right) _{y}\\
&=&0.
\end{eqnarray*}
The right-hand side of the above equation vanishes by virtue of the zero curvature equations (\ref{2.7}).
\end{proof}

It is well known (see e.g.\ \cite{bl,bla1,sem,sem2}) that whenever
$\mathfrak{g}$ admits a decomposition into two Lie subalgebras $\mathfrak{g}%
_{+}$ and $\mathfrak{g}_{-}$ such that
\[
\mathfrak{g}=\mathfrak{g}_{+}\oplus\mathfrak{g}_{-},\qquad[\mathfrak{g}_{\pm
},\mathfrak{g}_{\pm}]\subset\mathfrak{g}_{\pm},\qquad\mathfrak{g}_{+}%
\cap\mathfrak{g}_{-}=\emptyset,
\]
the operator
\begin{equation}
R=\frac{1}{2}(P_{+}-P_{-})=P_{+}-\frac{1}{2} \label{2.14}%
\end{equation}
where $P_{\pm}$ are projectors onto $\mathfrak{g}_{\pm}$, satisfies the
classical modified Yang--Baxter equation (\ref{2.3}) with $\alpha=\frac{1}{4}%
$, i.e., $R$ defined by (\ref{2.14}) is a classical $R$-matrix.

Next, let us specify the dependence of $L_{j}$ on $y$ via the so-called
Lax--Novikov equations (cf.\ \cite{bla3} and references therein)
\begin{equation}
[L_{j},\mathcal{L}]+(L_{j})_{y}=0,\qquad j\in\mathbb{N}. \label{2.11}%
\end{equation}
Then, upon applying (\ref{2.4}), (\ref{2.14}) and (\ref{2.11}), equations
(\ref{2.5}), (\ref{2.7}) and (\ref{2.12}) are readily seen to take
the following form:
\begin{equation}
(L_{s})_{t_{r}}=[B_{r},L_{s}],\qquad r,s\in\mathbb{N}, \label{2.15}%
\end{equation}%
\begin{equation}
(B_{r})_{t_{s}}-(B_{s})_{t_{r}}+[B_{r},B_{s}]=0, \label{2.16}%
\end{equation}%
\begin{equation}
\mathcal{L}_{t_{r}}=[B_{r},\mathcal{L}]+(B_{r})_{y},\qquad n,r\in\mathbb{N}
\label{2.17}%
\end{equation}
where $B_{i}=P_{+}L_{i}$.

Obviously, if upon the reduction to the case when all quantities are
independent of $y$ we put $\mathcal{L}=L_{n}$ for some $n\in\mathbb{N}$, then
the hierarchies (\ref{2.12}) boil down to hierarchies (\ref{2.5}) and the
Lax--Novikov equations (\ref{2.11}) reduce to (a part of) the commutativity
conditions (\ref{2.4}). In particular, if the bracket $[\cdot,\cdot]$ is such
that equations (\ref{2.12}) give rise to integrable systems in $d$ independent
variables, then equations (\ref{2.5}) yield integrable systems in $d-1$
independent variables.

A standard construction of a commutative subalgebra spanned by $L_{i}$ whose
existence by Theorem~\ref{t1} ensures commutativity of the flows (\ref{2.12}) is, in
the case of Lie algebras which admit an additional associative multiplication
$\circ$ which obeys the Leibniz rule
\begin{equation}
\operatorname{ad}_{a}(b\circ c)=\operatorname{ad}_{a}(b)\circ c+b\circ
\operatorname{ad}_{a}(c)\Leftrightarrow\lbrack a,b\circ c]=[a,b]\circ
c+b\circ\lbrack a,c], \label{2.1}%
\end{equation}
as follows: the commutative subalgebra in question is generated by fractional
powers of a given element $L\in\mathfrak{g}$, cf.\ e.g.\ \cite{bla1,sem2} and
references therein.

However, in our setting, when we no longer assume existence of an associative multiplication on $\mathfrak{g}$ which obeys (\ref{2.1}),
the construction from the preceding paragraph does not work
anymore. In order to circumvent this difficulty, instead of an
explicit construction of commuting $L_{i}$ we will \emph{impose}
the zero-curvature constraints (\ref{2.7}) on chosen elements
$L_{i}\in\mathfrak{g}$, $i\in\mathbb{N}$; it is readily seen that
in the setting of Sections~\ref{(3+1)} and~\ref{fc} we are
interested in, this can be done in a consistent fashion. By
Theorem~\ref{t1} this guarantees the commutativity of $L_{i}$ for
any $R$-matrix which obeys the classical modified Yang--Baxter
equation (\ref{2.3}) with $\alpha\neq0$.\looseness=-1

\section{The contact bracket}

\label{scb}

Consider a commutative and associative algebra $A$ of formal series in $p$
\begin{equation}
A\ni f=\sum_{i}u_{i}p^{i} \label{3.1}%
\end{equation}
with the standard multiplication%
\begin{equation}
f_{1}\cdot f_{2}\equiv f_{1}f_{2},\qquad f_{1},f_{2}\in A. \label{3.2}%
\end{equation}
The coefficients $u_{i}$ of these series are assumed to be smooth functions of
$x,y,z$ and infinitely many times $t_{1},t_{2},\dots$.

The contact bracket on $A$
will be denoted by $\{\cdot,\cdot\}_{C}$ and is defined in the same fashion as
in \cite{s}, that is,
\begin{equation}
\{f_{1},f_{2}\}_{C}=\displaystyle\frac{\partial f_{1}}{\partial p}%
\frac{\partial f_{2}}{\partial x}-p\frac{\partial f_{1}}{\partial p}%
\frac{\partial f_{2}}{\partial z}+f_{1}\frac{\partial f_{2}}{\partial
z}-(f_{1}\leftrightarrow f_{2}). \label{cb}%
\end{equation}
Notice that the variable $y$ is not involved in this bracket.

If we drop the dependence on $z$ then this bracket reduces to the
canonical Poisson bracket in one degree of freedom,
\begin{equation}
\{f_{1},f_{2}\}_{P,1}=\frac{\partial f_{1}}{\partial p}\frac{\partial f_{2}%
}{\partial x}-\frac{\partial f_{2}}{\partial p}\frac{\partial f_{1}}{\partial
x}, \label{p0}%
\end{equation}
where the variable $x$ is canonically conjugated to $p$.\looseness=-1

Note that $A$ is not a Poisson algebra as the contact bracket (\ref{cb}) does
not obey the Leibniz rule. However, it belongs to a more general class of the
so-called Jacobi algebras (see e.g.\ \cite{gra} and references therein for
further details on those) that obey the following generalization of the
Leibniz rule:\looseness=-1
\begin{equation}
\{f_{1}f_{2},f_{3}\}_{C}=\{f_{1},f_{3}\}_{C}f_{2}+f_{1}\{f_{2},f_{3}%
\}_{C}-f_{1}f_{2}\{1,f_{3}\}_{C}. \label{jlr}%
\end{equation}

More precisely, a Jacobi algebra is an associative commutative algebra (i.e.,
a vector space endowed with an associative commutative multiplication which is
distributive with respect to addition and compatible with multiplication by
elements of the ground field) which is further endowed with the Lie algebra
structure that obeys the generalized Leibniz rule (\ref{jlr}). If the unity
$1$ belongs to the center of the Lie algebra in question, then (\ref{jlr}) boils
down to the usual Leibniz rule and the algebra under study is then just a
Poisson algebra.

Now let $\mathcal{A}$ be a Lie algebra of formal series in two variables
$p_{x}$ and $p_{z}$ whose coefficients again depend on $x,y,z,t_{1}%
,t_{2},\dots$ with respect to the standard Poisson bracket in two degrees of
freedom:
\begin{equation}
\{h_{1},h_{2}\}_{P}=\frac{\partial h_{1}}{\partial p_{x}}\frac{\partial h_{2}%
}{\partial x}+\frac{\partial h_{1}}{\partial p_{z}}\frac{\partial h_{2}%
}{\partial z}-(h_{1}\leftrightarrow h_{2}). \label{3.3}%
\end{equation}
It is readily checked that we have \cite{s} a Lie algebra homomorphism from
$A$ to $\mathcal{A}$
\begin{equation}
\label{m}f(p,x,y,z,t_{1},t_{2},\dots)\mapsto\bar{f}=p_{z}f(p_{x}%
/p_{z},x,y,z,t_{1},t_{2},\dots).
\end{equation}
Note, however, that when we lift this homomorphism to the Jacobi algebra
homomorphism, we have
\[
\overline{f_{1}f_{2}}=\frac{1}{p_{z}}\bar{f}_{1}\bar{f}_{2}.
\]

It is now readily seen that in fact we have the Jacobi algebra
\emph{isomorphism}, given by (\ref{m}), that goes from the Jacobi algebra
$(A,\{,\}_{C},\cdot)$, defined via (\ref{3.1}), (\ref{3.2}) and (\ref{cb}), to
the Jacobi algebra $(\bar{A},\{,\}_{P},\circ)$ of formal series of the form
\begin{equation}
h=\sum_{i}u_{i}p_{x}^{i}p_{z}^{-i+1}, \label{3.4}%
\end{equation}
which is a subalgebra of $(\mathcal{A},\{,\}_{P},\circ)$, where
\begin{equation}
h_{1}\circ h_{2}=\frac{1}{p_{z}}h_{1}h_{2}. \label{3.5}%
\end{equation}
Notice that the bracket (\ref{3.3}) is not a Poisson bracket on the algebra
$(\bar{A},\{,\}_{P},\circ)$ as it does not obey the Leibniz rule with respect
to the multiplication (\ref{3.5}).\looseness=-1

To make contact with the $R$-matrix approach of Section~\ref{rma}, we identify
$\mathfrak{g}$ with $A$ and the bracket $[\cdot,\cdot]$ in $\mathfrak{g}$
with the contact bracket (\ref{cb}). As for the choice of the splitting of
$\mathfrak{g}$ into Lie subalgebras $\mathfrak{g}_{\pm}$ with $P_{\pm}$ being
projections onto the respective subalgebras, so $\mathfrak{g}_{\pm}=P_{\pm
}(\mathfrak{g})$, it is readily checked that we have two natural choices when
the $R$'s defined by (\ref{2.14}) satisfy the classical modified Yang--Baxter
equation (\ref{2.3}) and thus are $R$-matrices. These two choices
are\looseness=-1
\[
P_{+}=P_{\geqslant k},
\]
where $k=0$ or $k=1$, and by definition
\[
P_{\geqslant k}\left(  \sum\limits_{j=-\infty}^{\infty}a_{j}p^{j}\right)
=\sum\limits_{j=k}^{\infty}a_{j}p^{j}.
\]

Note that, in contrast with the (1+1)-dimensional systems associated with the
Poisson bracket (\ref{p0}) with one degree of freedom \cite{bla4}, the choice
of $k=2$, i.e., taking $P_{\geqslant2}$ for $P_{+}$, does not yield an
$R$-matrix on $A$ via (\ref{2.14}), that is, in this case $R$ defined via
(\ref{2.14}) does not satisfy (\ref{2.3}).

\section{Integrable (3+1)-dimensional infinite-component hierarchies and their
reductions}

\label{(3+1)}

Consider first the case of $k=0$ and the $n$th order Lax function from $A$
\begin{equation}
\mathcal{L}=u_{n}p^{n}+u_{n-1}p^{n-1}+\cdots+u_{0}+u_{-1}p^{-1}+\cdots
,\qquad n>0 \label{5.1}%
\end{equation}
and let
\begin{equation}
B_{m}\equiv P_{+}L_{m}=v_{m,m}p^{m}+v_{m,m-1}p^{m-1}+\cdots+v_{m,0},\qquad m>0
\label{5.2}%
\end{equation}
where $u_{i}=u_{i}(\vec{t},x,y,z)$, $v_{m,j}=v_{m,j}(\vec{t},x,y,z)$, and
$\vec{t}=(t_{1},t_{2},\dots)$.

Substituting $\mathcal{L}$ and $B_{m}$ into the zero-curvature Lax equations
\begin{equation}
\mathcal{L}_{t_{m}}=\{B_{m},\mathcal{L}\}_{C}+(B_{m})_{y} \label{5.3}%
\end{equation}
we obtain a hierarchy of infinite-component systems of the form
\begin{equation}
\begin{array}{rl}
(u_{r})_{t_{m}}  &  =X_{r}^{m}[u,v_{m}],\qquad r\leq n+m,\quad r\neq
0,\dots,m,\\
(u_{r})_{t_{m}}  &  =X_{r}^{m}[u,v_{m}]+(v_{m,r})_{y},\qquad r=0,\dots,m.
\end{array}\label{5.13}
\end{equation}
where in (\ref{5.13}) we put $u_{r}\equiv0$ for $r>n$ and
\begin{equation}%
\begin{array}
[c]{rcl}%
X_{r}^{m}[u,v_{m}] & = & \displaystyle\sum\limits_{s=0}^{m}[sv_{m,s}%
(u_{r-s+1})_{x}-(r-s+1)u_{r-s+1}(v_{m,s})_{x}\\[5mm]
&  & \quad-(s-1)v_{m,s}(u_{r-s})_{z}+(r-s-1)u_{r-s}(v_{m,s})_{z}],
\end{array}
\label{5.5}%
\end{equation}
for $r\leq m+n$, $u=(u_{n},u_{n-1},\dots)$ and $v_{m}=(v_{m,0},\dots,v_{m,m})$.
The fields $u_{r}$ for $r\leq n$ are dynamical variables while
equations for $n+m\geq r>n$ can be seen as nonlocal constraints on $u_{r}$
which define the variables $v_{m,s}$. The reader has to bear in mind that the
additional dependent variables $v_{m,s}$ are by construction related to each
other for different $m$ through the zero-curvature equations (\ref{2.16}).

Upon using the homomorphism (\ref{m}) we see that the hierarchy (\ref{5.13}) can
also be generated by
\[
\bar{\mathcal{L}}=u_{n}p_{x}^{n}p_{z}^{-n+1}+u_{n-1}p_{x}^{n-1}p_{z}%
^{-n+2}+\cdots+u_{0}p_{z}+u_{-1}p_{x}^{-1}p_{z}^{2}+\cdots,
\]
\[
\bar{B}_{m}=v_{m,m}p_{x}^{m}p_{z}^{-m+1}+v_{m,m-1}p_{x}^{m-1}p_{z}%
^{-m+2}+\cdots+v_{m,0}p_{z},
\]
and the Lax equations
\[
\bar{\mathcal{L}}_{t_{m}}=\{\bar{B}_{m},\bar{\mathcal{L}}\}_{P}+(\bar{B}%
_{m})_{y}%
\]
with the Lie bracket (\ref{3.3}). The same procedure can be applied to the
other examples given below, but in what follows we shall stick to
the contact bracket formalism for the sake of simplicity. Let us
also point out that using the contact bracket
$\{\cdot,\cdot\}_{C}$ and the algebra $A$ instead of $\bar{A}$ and
the Poisson bracket $\{\cdot,\cdot\}_{P}$ naturally leads to
nonisospectral Lax representations for systems written in the form
of zero-curvature equations like (\ref{2.12}) or (\ref{2.17})
with $[\cdot ,\cdot]$ being the contact bracket, cf.\ Theorem~1 of
\cite{s} for details.\looseness=-1

The first equation from the system (\ref{5.13}), i.e., the one for $r=n+m$, takes
the form
\[
(n-1)u_{n}(v_{m,m})_{z}-(m-1)v_{m,m}(u_{n})_{z}=0,
\]
and hence, for $n>1,m>1$, admits the constraint
\begin{equation}
v_{m,m}=(u_{n})^{\frac{m-1}{n-1}}. \label{5.6}%
\end{equation}
For $n=1$ the constraint in question takes the form $u_{1}=\mathrm{const}$.

The system (\ref{5.1})--(\ref{5.5}) has a natural constraint: $u_{n}%
=c_{n},\ v_{m,m}=c_{m,m}$, where $c_{n},c_{m,m}\in\mathbb{R}$. Then, if we put
$c_{n}=c_{m,m}=1$, we have
\begin{equation}
\mathcal{L}=p^{n}+u_{n-1}p^{n-1}+\cdots+u_{0}+u_{-1}p^{-1}+\cdots,\quad n>0,
\label{5.6a}%
\end{equation}%
\begin{equation}
B_{m}\equiv P_{+}L_{m}=p^{m}+v_{m,m-1}p^{m-1}+\cdots+v_{m,0},\qquad m>0
\label{5.6b}%
\end{equation}
and equations (\ref{5.3}) take the form (\ref{5.13}), where now $r<n+m$ and
\begin{equation}%
\begin{array}
[c]{rcl}%
X_{r}^{m}[u,v_{m}] & = & m(u_{r-m+1})_{x}-(m-1)(u_{r-m})_{z}\\[3mm]
&  & +\displaystyle\sum\limits_{s=0}^{m-1}[sv_{m,s}(u_{r-s+1})_{x}%
-(r-s+1)u_{r-s+1}(v_{m,s})_{x}\\[5mm]
&  & \quad-(s-1)v_{m,s}(u_{r-s})_{z}+(r-s-1)u_{r-s}(v_{m,s})_{z}].
\end{array}
\label{5.6c}%
\end{equation}
Again, the first equation from the system (\ref{5.13}), i.e., the one for
$r=n+m-1$, takes the form
\[
(n-1)(v_{m,m-1})_{z}-(m-1)(u_{n-1})_{z}=0,
\]
so the system under study for $n>1$ admits a further constraint
\begin{equation}
v_{m,m-1}=\frac{(m-1)}{(n-1)}u_{n-1}. \label{5.6d}%
\end{equation}

It is readily seen that for $n=1$ the constraint (\ref{5.6d}) should be replaced by
$u_{0}=\mathrm{const}$. Consider this case in more
detail.

Upon taking $u_{0}=0$, the Lax equation (\ref{5.3}) for
\begin{equation}
\mathcal{L}=p+u_{-1}p^{-1}+u_{-2}p^{-2}+\cdots,\ \label{5.7a}%
\end{equation}
and for $m=2$, with
\[
\ B_{2}=p^{2}+v_{1}p+v_{0},
\]
generates the following infinite-component system
\begin{align}
(v_{1})_{y}  &  =(v_{1})_{x}+(u_{-1})_{z},\nonumber\\
(v_{0})_{y}  &  =(v_{0})_{x}+(u_{-2})_{z}-2(u_{-1})_{x}+2u_{-1}(v_{1}%
)_{z},\nonumber\\
(u_{r})_{t_{2}}  &  =2(u_{r-1})_{x}-(u_{r-2})_{z}-(r+1)u_{r+1}(v_{0}%
)_{x}+v_{0}(u_{r})_{z}\label{5.7}\\
&  ~~\ \ +(r-1)u_{r}(v_{0})_{z}+v_{1}(u_{r})_{x}-ru_{r}(v_{1})_{x}%
+(r-2)u_{r-1}(v_{1})_{z},\nonumber
\end{align}
where $r<0$ and $v_{2,r}\equiv v_{r}$.

We have a natural $(2+1)$-dimensional reduction of (\ref{5.7}) when
$u_{j},v_{0}$ and $v_{1}$ are independent of $y$,
\begin{equation}%
\begin{array}
[c]{rcl}%
0 & = & (v_{1})_{x}+(u_{-1})_{z},\\
0 & = & (v_{0})_{x}+(u_{-2})_{z}-2(u_{-1})_{x}+2u_{-1}(v_{1})_{z},\\
(u_{r})_{t_{2}} & = & 2(u_{r-1})_{x}-(u_{r-2})_{z}-(r+1)u_{r+1}(v_{0}%
)_{x}+v_{0}(u_{r})_{z}\\
&  & +(r-1)u_{r}(v_{0})_{z}+v_{1}(u_{r})_{x}-ru_{r}(v_{1})_{x}+(r-2)u_{r-1}%
(v_{1})_{z},
\end{array}
\label{5.8}%
\end{equation}
another $(2+1)$-dimensional reduction
\begin{align}
(v_{1})_{y}  &  =(u_{-1})_{z},\nonumber\\
(v_{0})_{y}  &  =(u_{-2})_{z}+2u_{-1}(v_{1})_{z},\label{5.8a}\\
(u_{r})_{t_{2}}  &  =-(u_{r-2})_{z}+v_{0}(u_{r})_{z}+(r-1)u_{r}(v_{0}%
)_{z}+(r-2)u_{r-1}(v_{1})_{z},\nonumber
\end{align}
when $u_{j},v_{0}$ and $v_{1}$ are independent of $x$, and yet another $(2+1)$-dimensional reduction
\begin{align}
(v_{1})_{y}  &  =(v_{1})_{x},\nonumber\\
(v_{0})_{y}  &  =(v_{0})_{x}-2(u_{-1})_{x},\label{5.9}\\
(u_{r})_{t_{2}}  &  =2(u_{r-1})_{x}-(r+1)u_{r+1}(v_{0})_{x}+v_{1}(u_{r}%
)_{x}-ru_{r}(v_{1})_{x},\nonumber
\end{align}
when $u_{j},v_{0}$ and $v_{1}$ are independent of $z$.

Moreover, system (\ref{5.9}) admits a further
reduction $v_{1}=0$ to the form
\begin{align}
(v_{0})_{y}  &  =(v_{0})_{x}-2(u_{-1})_{x},\label{5.9a}\\
(u_{r})_{t_{2}}  &  =2(u_{r-1})_{x}-(r+1)u_{r+1}(v_{0})_{x}+v_{1}(u_{r}%
)_{x}.\nonumber
\end{align}
The system (\ref{5.9a}) reduces to $(1+1)$-dimensional Benney system
(cf.\ e.g.\ \cite{bl-pla, bla4})
\begin{equation}
(u_{r})_{t_{2}}=2(u_{r-1})_{x}-2(r+1)u_{r+1}(u_{-1})_{x},\quad r<0,
\label{5.10}%
\end{equation}
when $u_{i}$ are independent of \emph{both} $y$ and $z$, and we put $v_{0}=2u_{-1}$.

On the other hand, system (\ref{5.8a}) admits no reductions to $(1+1)$-dimensional systems. Note
that for systems (\ref{5.7})--(\ref{5.10}) there are no obvious
finite-component reductions.

For systems with the Lax functions (\ref{5.1}), (\ref{5.2}) and
(\ref{5.6a}), (\ref{5.6b}) we have $(2+1)$-dimensional and
$(1+1)$-dimensional reductions of the same types as above.

Now pass to the case of $k=1$, when $P_{+}=P_{\geqslant1}$, and consider the
general case when
\begin{align}
\mathcal{L}  &  =u_{n}p^{n}+u_{n-1}p^{n-1}+\cdots+u_{0}+u_{-1}p^{-1}+\dots,\quad n>0,\nonumber\\
B_{m}  &  =v_{m,m}p^{m}+v_{m,m-1}p^{m-1}+\cdots+v_{m,1}p,\quad m>0,
\label{5.27}%
\end{align}
from which we again obtain the hierarchies of infinite-component systems
\begin{equation}
\hspace*{-10mm}%
\begin{array}
[c]{rcl}%
(u_{r})_{t_{m}} & = & X_{r}^{m}[u,v_{m}],\qquad r\leq n+m,\quad r\neq
1,\dots,m,\\
(u_{r})_{t_{m}} & = & X_{r}^{m}[u,v_{m}]+(v_{m,r})_{y},\qquad r=1,\dots,m,
\end{array}
\label{5.28a}%
\end{equation}
where in (\ref{5.28a}) we put $u_{r}\equiv0$ for $r>n$ and
\begin{equation}%
\begin{array}
[c]{rcl}%
X_{r}^{m}[u,v_{m}] & = & \displaystyle\sum\limits_{s=1}^{m}[sv_{m,s}%
(u_{r-s+1})_{x}-(r-s+1)u_{r-s+1}(v_{m,s})_{x}\\[5mm]
&  & \quad-(s-1)v_{m,s}(u_{r-s})_{z}+(r-s-1)u_{r-s}(v_{m,s})_{z}],
\end{array}
\label{5.28b}%
\end{equation}
for $r\leq m+n$, $u=(u_{n},u_{n-1},\dots)$ and $v_{m}=(v_{m,1},\dots,v_{m,m})$.

For $n>1,m>1$ we again obtain the constraint (\ref{5.6}), and for $n=1$ the
constraint in question is replaced by $u_{1}=\mathrm{const}$.
Consider in more detail the simplest case when
\begin{equation}
\mathcal{L}=p+u_{0}+u_{-1}p^{-1}+\cdots\label{5.17}%
\end{equation}
and
\begin{equation}
B_{m}\equiv P_{+}L_{m}=v_{m,m-1}p^{m}+v_{m,m-2}p^{m-1}+\dots+v_{m,1}p,\qquad
m>1. \label{5.18}%
\end{equation}

The first flow for $m=2$, where we put $v_{2,r}\equiv v_{r}$ to
simplify writing, takes the form
\begin{align}
(v_{2})_{y}  &  =(v_{2})_{x}+u_{0}(v_{2})_{z}+v_{2}(u_{0})_{z},\nonumber\\
(v_{1})_{y}  &  =(v_{1})_{x}+u_{0}(v_{1})_{z}+v_{2}(u_{-1})_{z}+2u_{-1}%
(v_{2})_{z}-2v_{2}(u_{0})_{x},\nonumber\\
(u_{r})_{t_{2}}  &  =v_{1}(u_{r})_{x}-ru_{r}(v_{1})_{x}+(r-2)u_{r-1}%
(v_{1})_{z}+2v_{2}(u_{r-1})_{x}\label{5.22}\\
&  \quad-(r-1)u_{r-1}(v_{2})_{x}-v_{2}(u_{r-2})_{z}+(r-3)u_{r-2}(v_{2}%
)_{z}.\nonumber
\end{align}
We have a natural $(2+1)$-dimensional reduction of (\ref{5.22}) when
$u_{j},v_{1}$ and $v_{2}$ are independent of $y$,
\begin{equation}%
\begin{array}
[c]{rcl}%
0 & = & (v_{2})_{x}+u_{0}(v_{2})_{z}+v_{2}(u_{0})_{z},\\
0 & = & (v_{1})_{x}+u_{0}(v_{1})_{z}+v_{2}(u_{-1})_{z}+2u_{-1}(v_{2}%
)_{z}-2v_{2}(u_{0})_{x},\\
(u_{r})_{t_{2}} & = & v_{1}(u_{r})_{x}-ru_{r}(v_{1})_{x}+(r-2)u_{r-1}%
(v_{1})_{z}+2v_{2}(u_{r-1})_{x}\\
&  & -(r-1)u_{r-1}(v_{2})_{x}-v_{2}(u_{r-2})_{z}+(r-3)u_{r-2}(v_{2})_{z}.
\end{array}
\label{22.a}%
\end{equation}
On the other hand, if $u_{j},v_{1}$ and $v_{2}$ are independent of $x$, we
obtain from (\ref{5.22}) a $(2+1)$-dimensional system
\begin{align}
(v_{2})_{y}  &  =u_{0}(v_{2})_{z}+v_{2}(u_{0})_{z},\nonumber\\
(v_{1})_{y}  &  =u_{0}(v_{1})_{z}+v_{2}(u_{-1})_{z}+2u_{-1}(v_{2}%
)_{z},\label{22b}\\
(u_{r})_{t_{2}}  &  =(r-2)u_{r-1}(v_{1})_{z}-v_{2}(u_{r-2})_{z}+(r-3)u_{r-2}%
(v_{2})_{z}.\nonumber
\end{align}
Finally, if $u_{j},v_{1}$ and $v_{2}$ in (\ref{5.22}) are independent of $z$, we arrive at a $(2+1)$-dimensional system
\begin{align}
(v_{1})_{y}  &  =(v_{1})_{x},\nonumber\\
(v_{0})_{y}  &  =(v_{0})_{x}-2v_{1}(u_{0})_{x},\label{5.22d}\\
(u_{r})_{t_{2}}  &  =v_{0}(u_{r})_{x}-ru_{r}(v_{0})_{x}+2v_{1}(u_{r-1}%
)_{x}-(r-1)u_{r-1}(v_{1})_{x},\nonumber
\end{align}
where we made use of an admissible reduction $v_{2}=\mathrm{const}=1$, and if we make a further reduction $v_{1}=\mathrm{const}=1$, we obtain
\begin{align}
(v_{0})_{y}  &  =(v_{0})_{x}-2(u_{0})_{x},\label{22c}\\
(u_{r})_{t_{2}}  &  =2(u_{r-1})_{x}+v_{0}(u_{r})_{x}-ru_{r}(v_{0}%
)_{x}.\nonumber
\end{align}
If $u_{j},v_{1}$ and $v_{2}$ are independent of both $y$ and $z$, we can put
$v_{1}=2u_{0}$ and obtain
\begin{equation}
(u_{r})_{t_{2}}=2(u_{r-1})_{x}+2u_{0}(u_{r})_{x}-2ru_{r}(u_{0})_{x}.
\label{22d}%
\end{equation}
Finally, when $u_{j},v_{1}$ and $v_{2}$ are independent of both $y$ and $x$, we
have
\begin{equation}
(u_{r})_{t_{2}}=(r-2)u_{r-1}(v_{1})_{z}-v_{2}(u_{r-2})_{z}+(r-3)u_{r-2}%
(v_{2})_{z}, \label{22e}%
\end{equation}
where a reduction
\[
v_{2}=au_{0}^{-1},\quad v_{1}=-au_{-1}u_{0}^{-2},
\]
was performed, and $a\in\mathbb{R}$ is an arbitrary constant. Thus, in this case the system under study
is rational (rather than polynomial) in $u_0$.

\section{Finite-component reductions}\label{fc}
For $k=0$, in contrast with the simplest case (\ref{5.7a}), we do
have natural reductions to finite-component systems by putting
$u_{r}=0$ for $r<1$ or $r<0$ in (\ref{5.1}) and (\ref{5.6a}), i.e.,
consider the cases
\begin{equation}%
\begin{array}
[c]{rcl}%
\mathcal{L} & = & u_{n}p^{n}+u_{n-1}p^{n-1}+\cdots+u_{r}p^{r},\quad r=0,1,\\[3mm]
B_{m} & = & (u_{n})^{\frac{m-1}{n-1}}p^{m}+v_{m,m-1}p^{m-1}+\cdots+v_{m,0}%
\end{array}
\label{5.14}%
\end{equation}
and
\begin{equation}%
\begin{array}
[c]{rcl}%
\mathcal{L} & = & p^{n}+u_{n-1}p^{n-1}+\cdots+u_{r}p^{r},\quad r=0,1,\\[5mm]
B_{m} & = & p^{m}+\displaystyle\frac{(m-1)}{(n-1)}u_{n-1}p^{m-1}+\cdots+v_{m,0}.
\end{array}
\label{5.14a}%
\end{equation}
The case (\ref{5.14a}) for $r=0$ was considered for the first time in
\cite{s}. Notice that in (\ref{5.14}) and (\ref{5.14a}) for $r=0$ we have
$\mathcal{L}=B_{n}$, and hence the variable $y$ can be identified with $t_{n}$.
Then equations (\ref{5.3}) coincide with the zero-curvature equations
(\ref{2.16}) and the Lax--Novikov equations (\ref{2.11}) reduce to equations
(\ref{2.15}).\looseness=-1

The structure of the said finite-component reductions is best revealed in the matrix
form of the system (\ref{5.13}). For the reduction (\ref{5.14a}) and $n\geq m$
we obtain
\begin{equation}%
\begin{array}
[c]{rcl}%
0 & = & A_{1}^{m}(u)(V_{m})_{z}+A_{2}^{m}(u)(V_{m})_{x}+A_{3}^{m}%
(v)(U_{m})_{z}+A_{4}^{m}(v)(U_{m})_{x},\\
(U_{n})_{t_{m}} & = & A_{1}^{n}(u)(V_{n})_{z}+A_{2}^{n}(u)(V_{n})_{x}%
+A_{3}^{n}(v)(U_{n})_{z}+A_{4}^{n}(v)(U_{n})_{x}+(V_{n})_{y},
\end{array}
\label{5.15}%
\end{equation}
where
\[
U_{m}=(u_{n-m},\dots,u_{n-1})^{T},\quad V_{m}=(v_{m,0},\dots,v_{m,m-1})^{T},
\]%
\[
V_{n}=\underset{n}{\underbrace{(v_{m,0},\dots,v_{m,m-1},0,\dots,0)}}^{T},\quad
U_{n}=(u_{0},\dots,u_{n-1})^{T},
\]
$A_{i}^{m}$ and $A_{i}^{n}$ are respectively $m\times m$ and $n\times n$
square matrices, and, as usual, the superscript $T$ indicates the transposed matrix.
The entries of the matrices in question are linear in the fields
$u_{i}$ and $v_{m,s}$.

On the other hand, for $n<m$ we have
\begin{equation}%
\begin{array}
[c]{rcl}%
0 & = & B_{1}^{m}(u)(V_{m})_{z}+B_{2}^{m}(u)(V_{m})_{x}+B_{3}^{m}%
(v)(U_{m})_{z}+B_{4}^{m}(v)(U_{m})_{x}+(V_{m,n})_{y},\\[3mm]%
(U_{n})_{t_{m}} & = & B_{1}^{n}(u)(V_{n})_{z}+B_{2}^{n}(u)(V_{n})_{x}%
+B_{3}^{n}(v)(U_{n})_{z}+B_{4}^{n}(v)(U_{n})_{x}+(V_{n})_{y},
\end{array}
\label{5.16}%
\end{equation}
where
\[
V_{m}=(v_{m,0},\dots,v_{m,m-1})^{T},\quad U_{m}=\underset{m}{\underbrace
{(u_{0},\dots,u_{n-1},0,\dots,0)}}^{T},
\]%
\[
V_{m,n}=\underset{m}{\underbrace{(v_{m,n},\dots,v_{m,m-1},0,\dots,0)}}^{T},
\]%
\[
U_{n}=(u_{0},\dots,u_{n-1})^{T},\quad V_{n}=(v_{m,0},\dots,v_{m,n-1})^{T}.
\]
The structure of the matrices $B_{i}^{(j)}$ is essentially the same as that of
the matrices $A_{i}^{(j)}$ above.

Another class of natural reductions
to finite-component systems
arises for $k=1$, if we put%
\begin{equation}%
\begin{array}
[c]{rcl}%
\mathcal{L} & = & u_{n}p^{n}+u_{n-1}p^{n-1}+\cdots+u_{r}p^{r},\quad
r=1,0,-1,\cdots\\
B_{m} & = & (u_{n})^{\frac{m-1}{n-1}}p^{m}+v_{m,m-1}p^{m-1}+\cdots+v_{m,1}p
\end{array}
\label{5.22a}%
\end{equation}
or%
\begin{align}
\mathcal{L}  &  =p+u_{0}+u_{-1}p^{-1}+\cdots+u_{r}p^{r},\quad r=0,1,-1,\dots
\label{5.23}\\
B_{m}  &  =v_{m,m}p^{m}+v_{m,m-1}p^{m-1}+\dots+v_{m,1}p,\qquad m>1.\nonumber
\end{align}
For instance, let
\begin{equation}
\mathcal{L}=p+u_{0}+u_{-1}p^{-1} \label{5.24}%
\end{equation}
and, with a slight variation of the earlier notation, put
\[
B_{2}=v_{2}p^{2}+v_{1}p,\quad B_{3}=w_{3}p^{3}+w_{2}p^{2}+w_{1}p.
\]
The member of the hierarchy associated with $B_{2}$ reads%
\begin{align}
(u_{-1})_{t_{2}}  &  =u_{-1}(v_{1})_{x}+v_{1}(u_{-1})_{x},\nonumber\\
(u_{0})_{t_{2}}  &  =-2u_{-1}(v_{1})_{z}+v_{1}(u_{0})_{x}+u_{-1}(v_{2}%
)_{x}+2v_{2}(u_{-1})_{x},\nonumber\\
(v_{1})_{y}  &  =(v_{1})_{x}+2u_{-1}(v_{2})_{z}+v_{2}(u_{-1})_{z}+u_{0}%
(v_{1})_{z}-2v_{2}(u_{0})_{x},\label{ut2}\\
(v_{2})_{y}  &  =(v_{2})_{x}+u_{0}(v_{2})_{z}+v_{2}(u_{0})_{z},\nonumber
\end{align}
and the one associated with $B_{3}$ has the form
\begin{align}
(u_{-1})_{t_{3}}  &  =u_{-1}(w_{1})_{x}+w_{1}(u_{-1})_{x},\nonumber\\
(u_{0})_{t_{3}}  &  =w_{1}(u_{0})_{x}-2u_{-1}(w_{1})_{z}+u_{-1}(w_{2}%
)_{x}+2w_{2}(u_{-1})_{x},\nonumber\\
(w_{1})_{y}  &  =(w_{1})_{x}+w_{2}(u_{-1})_{z}-u_{-1}(w_{3})_{x}-2w_{2}%
(u_{0})_{x}+2u_{-1}(w_{2})_{z}\nonumber\\
&  +u_{0}(w_{1})_{z}-3w_{3}(u_{-1})_{x},\label{ut3}\\
(w_{2})_{y}  &  =(w_{2})_{x}-3w_{3}(u_{0})_{x}+2w_{3}(u_{-1})_{z}+w_{2}%
(u_{0})_{z}+u_{0}(w_{2})_{z}+2u_{-1}(w_{3})_{z},\nonumber\\
(w_{3})_{y}  &  =(w_{3})_{x}+u_{0}(w_{3})_{z}+2w_{3}(u_{0})_{z},\nonumber
\end{align}

Commutativity of the flows associated with $t_{2}$ and $t_{3}$,
i.e.,
\[
\left((u_{i})_{t_{2}}\right)_{t_{3}}=\left((u_{i})_{t_{3}}\right)_{t_{2}},\quad i=0,1,
\]
can be readily checked using the set of relations
\begin{equation}%
\begin{array}
[c]{rcl}%
(v_{1})_{z} & = & -\displaystyle\frac{v_{2}}{w_{3}}(w_{3})_{x}-\frac
{v_{2}w_{2}}{4w_{3}^{2}}(w_{3})_{z}+\frac{v_{2}}{2w_{3}}(w_{2})_{z}+\frac
{3}{2}(v_{2})_{x},\\[5mm]%
(v_{2})_{z} & = & \displaystyle\frac{v_{2}}{2w_{3}}(w_{3})_{z},\\[5mm]%
(w_{1})_{t_{2}} & = & v_{1}(w_{1})_{x}-w_{1}(v_{1})_{x}+(v_{1})_{t_{3}%
},\\[5mm]%
(w_{2})_{t_{2}} & = & v_{1}(w_{2})_{x}-w_{1}(v_{2})_{x}+2v_{2}(w_{1}%
)_{x}-2w_{2}(v_{1})_{x}+(v_{2})_{t_{3}},\\[5mm]%
(w_{3})_{t_{2}} & = & \displaystyle\frac{v_{2}w_{2}}{2w_{3}}(w_{2})_{z}%
-\frac{w_{2}}{2}(v_{2})_{x}-\frac{v_{2}w_{2}^{2}}{4w_{3}^{2}}(w_{3})_{z}%
+\frac{(v_{1}w_{3}-v_{2}w_{2})}{w_{3}}(w_{3})_{x}\\[5mm]
&  & -v_{2}(w_{1})_{z}+2v_{2}(w_{2})_{x}-3w_{3}(v_{1})_{x},\nonumber
\end{array}
\label{zcrt2t3}%
\end{equation}
which is equivalent to the zero-curvature equation
\begin{equation}
(B_{2})_{t_{3}}-(B_{3})_{t_{2}}+\{B_{2},B_{3}\}_{C}=0. \label{zcr-t2t3}%
\end{equation}
Note that the compatibility conditions
\[
\left((v_{i})_{y}\right)_{z}=\left((v_{i})_{z}\right)_{y},\quad
i=1,2,
\]
are also satisfied by virtue of (\ref{ut2}) and (\ref{zcr-t2t3}).

When $u_{i}$ and $v_{j}$ are independent of $z$ we obtain $(2+1)$-dimensional
systems with additional constraints $v_{2}=\mathrm{const}=1$, $w_{3}=\mathrm{const}=1$
\begin{align}
(u_{-1})_{t_{2}}  &  =u_{-1}(v_{1})_{x}+v_{1}(u_{-1})_{x},\nonumber\\
(u_{0})_{t_{2}}  &  =v_{1}(u_{0})_{x}+2(u_{-1})_{x},\label{100}\\
(v_{1})_{y}  &  =(v_{1})_{x}-2(u_{0})_{x},\nonumber
\end{align}
and
\begin{align}
(u_{-1})_{t_{3}}  &  =u_{-1}(w_{1})_{x}+w_{1}(u_{-1})_{x},\nonumber\\
(u_{0})_{t_{3}}  &  =w_{1}(u_{0})_{x}+u_{-1}(w_{2})_{x}+2w_{2}(u_{-1}%
)_{x},\nonumber\\
(w_{1})_{y}  &  =(w_{1})_{x}-3(u_{-1})_{x}-2w_{2}(u_{0})_{x},\label{101}\\
(w_{2})_{y}  &  =(w_{2})_{x}-3(u_{0})_{x}.\nonumber
\end{align}

When $u_{i}$ and $v_{j}$ are independent of $x$ we obtain other $(2+1)$-dimensional systems
making use of a naturally arising extra constraint $u_{-1}=1$, namely,
\begin{align}
(u_{0})_{t_{2}}  &  =-2(v_{1})_{z},\nonumber\\
(v_{1})_{y}  &  =2(v_{2})_{z}+u_{0}(v_{1})_{z},\label{102}\\
(v_{2})_{y}  &  =(u_{0}v_{2})_{z}\nonumber
\end{align}
and
\begin{align}
(u_{0})_{t_{2}}  &  =-2(w_{1})_{z},\nonumber\\
(w_{1})_{y}  &  =2(w_{2})_{z}+u_{0}(w_{1})_{z},\nonumber\\
(w_{2})_{y}  &  =2(w_{3})_{z}+(u_{0}w_{2})_{z},\label{103}\\
(w_{3})_{y}  &  =u_{0}(w_{3})_{z}+2w_{3}(u_{0})_{z}.\nonumber
\end{align}

Further reduction of (\ref{100}) and (\ref{101}) by assuming that $u_i$, $v_j$ and $w_k$ are independent of $y$ leads to
$(1+1)$-dimensional systems of the form
\begin{align}
(u_{-1})_{t_{2}} &  =2(u_{-1}u_{0})_{x},\nonumber\\
(u_{0})_{t_{2}} &  =2(u_{-1}+u_{0}^{2})_{x},\label{104}%
\end{align}
with the constraint $v_{1}=2u_{0}$, and
\begin{align}
(u_{-1})_{t_{3}} &  =3(u_{-1}u_{0}^{2}+u_{-1}^{2})_{x},\nonumber\\
(u_{0})_{t_{3}} &  =(u_{0}^{3}+6u_{0}u_{-1})_{x},\label{105}%
\end{align}
with constraints
\[
w_{2}=3u_{0},\qquad w_{1}=3u_{0}^{2}+3u_{-1}.
\]

Likewise, the reduction of (\ref{102}) and (\ref{103}) by assuming that $u_i$, $v_j$ and $w_k$
are independent of $y$ leads to $(1+1)$-dimensional systems of the form
\begin{equation}
(u_{0})_{t_{2}}=2(u_{0}^{-2})_{z}, \label{106}%
\end{equation}
and
\begin{equation}
(u_{0})_{t_{3}}=-6(u_{0}^{-4})_{z}, \label{107}%
\end{equation}
as we have
\[
v_{2}=u_{0}^{-1},\quad v_{1}=-u_{0}^{-2},\quad w_{3}=u_{0}^{-2},\quad w_{2}%
=-2u_{0}^{-3},\quad w_{1}=3u_{0}^{-4}.
\]

The simplest nontrivial example of Lax pair (\ref{5.22a}) is given by
\begin{align*}
\mathcal{L}  &  =u_{3}p^{3}+u_{2}p^{2}+u_{1}p,\\
B_{2}  &  =v_{2}p^{2}+v_{1}p,
\end{align*}
and the associated system reads
\begin{equation}
\hspace*{-10mm}%
\begin{array}
[c]{rcl}%
0 & = & 2u_{3}(v_{2})_{z}-v_{2}(u_{3})_{z},\\
0 & = & u_{2}(v_{2})_{z}-v_{2}(u_{2})_{z}+2u_{3}(v_{1})_{z}+2v_{2}(u_{3})_{x}-3u_{3}(v_{2})_{x}\nonumber\\
(u_{3})_{t_{2}} & = & v_{1}(u_{3})_{x}+2v_{2}(u_{2})_{x}-2u_{2}(v_{2})_{x}-3u_{3}(v_{1})_{x}-v_{2}(u_{1})_{z}+u_{2}(v_{1})_{z},\\
(u_{2})_{t_{2}} & = & (v_{2})_{y}+v_{1}(u_{2})_{x}+2v_{2}(u_{1})_{x}-2u_{2}(v_{1})_{x}-u_{1}(v_{2})_{x},\\
(u_{1})_{t_{2}} & = & (v_{1})_{y}+v_{1}(u_{1})_{x}-u_{1}(v_{1})_{x}.
\end{array}
\label{5.25}%
\end{equation}
Here we have not yet imposed the constraint (\ref{5.6}).

The first two of the above equations impose constraints on the
`non-dynamical' fields $v_{1}$ and $v_{2}$.
The first of these constraints is satisfied once we impose (\ref{5.6}), i.e.,
$v_{2}=(u_{3})^{\frac{1}{2}}$, and then the second one boils down to
\[
(v_{1})_{z}=\left[  \frac{1}{2}u_{2}(u_{3})^{-\frac{1}{2}}\right]
_{z}-\left[\frac{1}{2}(u_{3})^{\frac{1}{2}}\right]_{x}.
\]

Assuming that $u_i$ and $v_j$ no longer depend on $z$ naturally leads to further constraints
\[
v_{2}=\mathrm{const}=1,\quad u_{3}=\mathrm{const}=1,\quad v_{1}=\frac{2}{3}u_{2}%
\]
and then we obtain an evolutionary system
\begin{align}
(u_{2})_{t_{2}}  &  =2(u_{1})_{x}-\frac{2}{3}u_{2}(u_{2})_{x},\nonumber\\
(u_{1})_{t_{2}}  &  =\frac{2}{3}[(u_{2})_{y}+u_{2}(u_{1})_{x}-u_{1}(u_{2}%
)_{x}]. \label{5.26}%
\end{align}
On the other hand, assuming that $u_i$ and $v_j$ no longer depend on $x$ yields
\begin{align}
(u_{3})_{t_{2}}  &  =u_{2}(v_{1})_{z}-v_{2}(u_{1})_{z},\nonumber\\
(u_{2})_{t_{2}}  &  =(v_{2})_{y},\label{5.27a}\\
(u_{1})_{t_{2}}  &  =(v_{1})_{y},\nonumber
\end{align}
where we have
\[
v_{2}=(u_{3})^{\frac{1}{2}},\quad v_{1}=\frac{1}{2}u_{2}(u_{3})^{-\frac{1}{2}%
}.
\]

The reduction of (\ref{5.26}) and (\ref{103}) by assuming that
the dependent variables involved are independent of $y$ leads to
a $(1+1)$-dimensional system
\begin{align}
(u_{2})_{t_{2}}  &  =2(u_{1})_{x}-\frac{2}{3}u_{2}(u_{2})_{x},\nonumber\\
(u_{1})_{t_{2}}  &  =\frac{2}{3}[u_{2}(u_{1})_{x}-u_{1}(u_{2})_{x}],
\label{5.30}%
\end{align}
while for (\ref{5.27a}) we are naturally led to imposing the constraints
\[
u_{1}=\mathrm{const}=0,\quad u_{2}=\mathrm{const}=1,
\]
and then we obtain the equation
\begin{equation}
(u_{3})_{t_{2}}=\frac{1}{2}\left((u_{3})^{-\frac{1}{2}}\right)_{z}.
\label{5.31}%
\end{equation}

In closing note that it would be interesting to find out whether the
hierarchies presented in this section could be reproduced using the recursion
operators in spirit of \cite{bl, m-s0,m-s,as,sro} and references therein.\looseness=-1

\subsection*{Acknowledgments}

The research of AS was supported in part by the Ministry of Education, Youth
and Sports of the Czech Republic (M\v{S}MT \v{C}R) under RVO funding for
I\v{C}47813059, and by the Grant Agency of the Czech Republic (GA \v{C}R)
under grant P201/12/G028.

AS gratefully acknowledges the warm hospitality extended to him in the course
of his visits to the Adam Mickiewicz University in Pozna\'n.

The authors would like to thank B.M. Szablikowski for helpful comments.


\end{document}